\documentclass[11pt]{article}

\usepackage[english]{babel}
\usepackage{amsmath,amsthm}
\usepackage{amsfonts}
\usepackage{graphicx}

\newtheorem{theorem}{Theorem}[section]

\newtheorem{corollary}{Corollary}[section]
\newtheorem{proposition}{Proposition}[section]
\newtheorem{lemma}{Lemma}[section]
\newtheorem{definition}{Definition}[section]

\newtheorem{example}{Example}[section]

\newcommand{\bdes}{\begin{description}}
\newcommand{\edes}{\end{description}}

\newcommand{\bal}{\begin{align}}
\newcommand{\eal}{\end{align}}

\newcommand{\bnum}{\begin{enumerate}}
\newcommand{\enum}{\end{enumerate}}

\newcommand{\bit}{\begin{itemize}}
\newcommand{\eit}{\end{itemize}}

\newcommand{\bea}{\begin{eqnarray}}
\newcommand{\eea}{\end{eqnarray}}

\newcommand{\bsry}{\begin{subarray}}
\newcommand{\esry}{\end{subarray}}

\newcommand{\bca}{\begin{cases}}
\newcommand{\eca}{\end{cases}}

\newcommand{\bcen}{\begin{center}}
\newcommand{\ecen}{\end{center}}

\newcommand{\bbm}{\begin{bmatrix}}
\newcommand{\ebm}{\end{bmatrix}}

\newcommand{\bmx}{\begin{matrix}}
\newcommand{\emx}{\end{matrix}}

\newcommand{\bpm}{\begin{pmatrix}}
\newcommand{\epm}{\end{pmatrix}}

\newcommand{\btab}{\begin{tabular}}
\newcommand{\etab}{\end{tabular}}

\newcommand{\B}{\mathcal{B}}

\newcommand{\ijm}{{i_1\cdots i_m j_1 \cdots j_m}}
\newcommand{\jim}{{j_1\cdots j_m i_1 \cdots i_m}}

\newcommand{\cH}{\mathcal{H}}
\newcommand{\bH}{\mathbb{H}}

\newcommand{\cpx}{\mathbb{C}}
\newcommand{\lmd}{\lambda}

\newcommand{\reff}[1]{(\ref{#1})}
\newcommand{\mc}[1]{\mathcal{#1}}
\DeclareMathOperator{\rank}{rank}
\usepackage{mathrsfs}
\newcommand{\be}{\begin{equation}}
\newcommand{\ee}{\end{equation}}
\newcommand{\baray}{\begin{array}}
\newcommand{\earay}{\end{array}}

\def\re{\mathbb{R}}
\def\cpx{\mathbb{C}}

\def\be{\beta}

\def\nn{\nonumber}

\begin{document}
\title{\bf  Hermitian tensor and quantum mixed state\thanks{This work is supported by the National Natural Science Foundation of China (No. 11871472).}}
\author{{Guyan Ni\thanks{Corresponding author. \newline {\it $\mathrm{\ \ \ \ }$
E-mail address}: guyan-ni@163.com(Guyan Ni)} %
 }\\
 {\small\it Department of Mathematics, National University of
Defense Technology,}\\ {\small\it Changsha, Hunan 410073, China.}\\
\\ \newline
\begin{tabular}{p{14.8cm}} \hline\\ {\small {\bf Abstract:}
An order $2m$ complex tensor $\cH$ is said to be Hermitian if \[\mathcal{H}_\ijm=\mathcal{H}_\jim ^*\mathrm{\ for\ all\ }\ijm .\] It can be regarded as an extension of Hermitian matrix to higher order.
A Hermitian tensor is also seen as a representation of a quantum mixed state. Motivated by the separability discrimination of quantum states, we investigate  properties of Hermitian tensors including: unitary similarity relation, partial traces, nonnegative Hermitian tensors, Hermitian eigenvalues, rank-one Hermitian decomposition and positive Hermitian decomposition, and their applications to quantum states.
\medskip%
$~~$\newline{\it Keywords:} Hermitian tensor; tensor decomposition; tensor eigenvalue;
quantum mixed state.
$~~$\newline{ MSC2010}: {15A18, 15A69, 46B28, 81P40}} \\
\hline
\end{tabular}
}
\date{}
\maketitle

\vskip 2mm

\section{Introduction}
A fundamental and also important problem in quantum physics is to detect whether a state is separable or not, especially for quantum mixed states. Around this problem, the paper investigates properties of Hermitian tensors and their application to quantum states.

An $m$th-order complex tensor denoted by $\mathcal{A}=(\mathcal{A}_{i_1...i_m})\in \mathbb{F}^{n_1\times\ldots\times n_m}$ $(\mathbb{F}=\mathbb{R}$ or $\mathbb{C})$ is a multi-array consisting of numbers $\mathcal{A}_{i_1...i_m}\in\mathbb{F}$ for all $i_k\in [n_k]$ and $k\in [m] $, where $[n]:=\{1, ... , n\}$. If $ \mathbb{F}=\mathbb{R}(\mathrm{or}\ \mathbb{C})$ then $\mathcal{A}$ is called a real (or complex) tensor \cite{QCC2018,QL2017}.  Tensor is the extension of matrix to higher order. As an extension of symmetric matrices, a tensor $\mathcal{S}=(\mathcal{S}_{i_1...i_m})\in \mathbb{F}^{n\times\cdot\cdot\cdot\times n}$ is called symmetric \cite{Comon2008} if its entries $\mathcal{S}_{i_1...i_m}$ are invariant under any permutation  operator $P$ of $\{1,..., m\}$, i.e.
\begin{equation}
	\mathcal{S}_{i_1\ldots i_m}=\mathcal{S}_{P[i_1...i_m]},
\end{equation}
where $P[i_1...i_m]:= [i_{P[1]}...i_{P[m]}]$.

Lots of study have been conducted regarding properties of tensors such as tensor eigenvalues \cite{qi05,L05,FanNZh2018}, the best rank-one approximation \cite{ZLQ12,NieW2014,CCW17}, tensor rank \cite{Comon2008,landsberg2010}, tensor and symmetric tensor decomposition  \cite{nie2017low,Brachat2010,nie2017generating}, symmetric tensor \cite{Chang2009,LuoQY2015}, nonnegative tensor \cite{HuHQ2014}, copositive tensors \cite{QiL2013}, and completely positive tensor \cite{QiXX2014,FanZh2017,ZhouFan2018}. Many tensor computation methods are also proposed including tensor eigenvalue computation \cite{QWW09,kdm11,CuiDN2014,HCD15,CHL2016,YYXSZ16}, tensor system solution \cite{LiNg2015,DingWei2016,HanL2017,LiXX2017,XieJW2018,DuQZC2018,LingYHQ2019}, and tensor decomposition \cite{FJL18}.

In the study of matrices, symmetric matrices and Hermitian matrices are playing significant roles \cite{Horn1990}. Similarly, Hermitian tensor as an extension of Hermitian matrices is defined as follows.

\begin{definition}
\label{def:hermtensor}\rm
A $2m$th-order tensor $\mathcal{H}=(\mathcal{H}_{i_1...i_m j_1...j_m})\in \mathbb{C}^{n_1\times\cdots\times n_m\times n_1\times\cdots \times n_m} $ is called a {\it Hermitian tensor} if $\mathcal{H}_{i_1...i_m j_1...j_m}=\mathcal{H}^*_{j_1...j_m i_1...i_m}$ for every $i_1, ..., i_m$ and  $j_1, ..., j_m$, where $x^*$ denotes the complex conjugate of $x$. A Hermitian tensor $\mathcal{H}$ is called a {\it symmetric Hermitian tensor} if $ n_1=\cdots=n_m $ and its entries $\mathcal{H}_{i_1...i_m j_1...j_m}$ are invariant under any permutation operator $P$ of $\{1,...,m\}$ with
$\mathcal{H}_{i_1...i_m j_1...j_m}=\mathcal{H}_{P[i_1...i_m] P[j_1...j_m]}.$
\end{definition}

In \cite[Definition 3.7]{JLZ2016}, Jiang et al. call a Hermitian tensor $\cH$ as a {\it conjugate partial-symmetric tensor} if $ n_1=\cdots=n_m $ and its entries $\mathcal{H}_{i_1...i_m j_1...j_m}$ are invariant under any permutation operators $P$ and $Q$ of $\{1,...,m\}$ with
$\mathcal{H}_{i_1...i_m j_1...j_m}=\mathcal{H}_{P[i_1...i_m] Q[j_1...j_m]}.$ Hence, a conjugate partial-symmetric tensor is a special symmetric Hermitian tensor.

The space of all Hermitian tensors $\mathcal{H} \in \mathbb{C}^{n_1\times\cdots\times n_m\times n_1\times\cdots \times n_m}$ is denoted by $\bH {[n_1, \ldots, n_m]}$ for convenience. Generally an order $2m$ tensor $\cH$ is called Hermitian if there is a permutation $P[1,2,\cdots,2m]=[p_1,...,p_m, q_1,...,q_m]$ such that such that
\begin{equation}
    \B_{i_1\ldots i_{2m}}=\cH_{i_{p_1}\ldots i_{p_m} j_{q_1}\ldots j_{q_m}}
\end{equation}
$\B$ is a Hermitian tensor defined in (\ref{def:hermtensor}). The general Hermitian tensor can be transformed into a usual Hermitian tensor by index permutation. It suffices to study the usual Hermitian tensor in (\ref{def:hermtensor}), so our paper is focusing on the Hermitian tensor defined in (\ref{def:hermtensor}).

Complex tensors and Hermitian tensors play important roles in quantum physics research.
 An $m$-partite pure state $|\psi\rangle$ of a composite quantum system can be regarded as a normalized element in a Hilbert tensor product space $\mathbb{C}^{n_1\times\cdots\times n_m}$. The pure state $|\psi\rangle$ is denoted as
$$
|\psi\rangle =\sum_{i_1,\cdots,i_m=1}^{n_1,\cdots,n_m}\chi_{i_1\cdots i_m}|e_{i_1}^{(1)}\cdots e_{i_m}^{(m)}\rangle,
$$
where $\chi_{i_1\cdots i_m}\in \mathbb{C}$, $\{|e^{(k)}_{i_k}\rangle: i_k=1,2,\cdots, n_k\}$ is an orthonormal basis of $\mathbb{C}^{n_k}$. Hence, a pure state is uniquely corresponding to a complex tensor $\chi=(\chi_{i_1\cdots i_m})$ under a given orthonormal basis \cite{NQB14}. Furthermore, one can obtain the geometric measure of a pure state by computing the U(US)-eigenvalues of its corresponding complex tensor \cite{HNZ16,QZN2017}. There are many results on the computation of U(US)-eigenvalue \cite{NB16,CQW17,NZZ2017,CQWZ18} and complex tensor research \cite{JLZ2016,hqz12,ZhFanW2018}.

Similarly, for a quantum mixed state $\rho$, its density matrix is always written as
$$\rho=\sum_{i=1}^k \lambda_i |\psi_i \rangle \langle \psi_i |,$$
where $\lambda_i>0$ and $\sum_{i=1}^k \lambda_i=1$, $|\psi_i \rangle$ is a pure state and $\langle\psi_i |$ is the complex conjugate transpose of $|\psi_i \rangle$. Hence, the density matrix of $\rho$ is also uniquely corresponding to a Hermitian tensor $\mathcal{H}\in \bH {[n_1, \ldots, n_m]} $ with
$$\mathcal{H}=\sum_{i=1}^k \lambda_i \chi^{(i)}\otimes \chi^{(i)*},$$
where $\chi^{(i)} $ is the corresponding complex tensor of the state $|\psi_i \rangle $.



Motivated by the separability discrimination of quantum states, we study properties of Hermitian tensors and their applications. The paper is organized as follows.
Section \ref{sc:sbp} introduce tensor operations and derive properties including unitary similarity relation, and invariant properties under unitary transformation. Section \ref{sc:PartTT} studies partial traces of Hermitian tensors. Nonnegativity and Hermitian Eigenvalues are discussed in Section \ref{sc:nhe}. Section \ref{sc:R1HTD} proposes the rank-one Hermitian decomposition, proves the existence of Hermitian decomposition for Hermitian tensors, and discusses properties of positive Hermitian tensors. Section \ref{sec:mixedstates} is the theoretical application to quantum mixed state.


\section{Tensor operation and properties}
\label{sc:sbp}

 Let $\mathcal{A}, \mathcal{B}\in \bH {[n_1, \ldots, n_m]}$ be Hermitian tensors. Their inner product is defined as
\begin{equation}\label{Eq:tensorinnerproduct}
    \langle\mathcal{A},\mathcal{B}\rangle:= \sum_{i_1,\cdots,i_m,j_1,\cdots,j_m=1}^{n_1,\cdots,n_m,n_1,\cdots,n_m} \mathcal{A}^*_{i_1 \cdots i_m j_1 \cdots j_m} \mathcal{B}_{i_1 \cdots i_m j_1 \cdots j_m},
\end{equation}
 the Frobinius norm of tensor $\mathcal{A}$ is defined as
\begin{eqnarray}
 \label{Eq:F-norm}  ||\mathcal{A}||_F &:=& \sqrt{\langle\mathcal{A},\mathcal{A}\rangle}
\end{eqnarray}
and the matrix trace of tensor $\mathcal{A}$ is defined as
\begin{equation}\label{Eq:matrixtrace}
    \mathrm{Tr}_M \mathcal{A}:= \sum_{i_1,\cdots,i_m=1}^{n_1,\cdots,n_m} \mathcal{A}_{i_1 \cdots i_m i_1 \cdots i_m}.
\end{equation}

For vectors  $u_1 \in \cpx^{n_1}, \ldots,
u_m \in \cpx^{n_m}$, the rank-$1$ Hermitian tensor is denoted by
\[
\otimes_{i=1}^mu_i\otimes_{j=1}^m u_j^* := \,
u_1 \otimes \cdots \otimes u_m \otimes
u_1^* \otimes \cdots \otimes u_m^*.
\]
So $\otimes_{i=1}^mu_i\otimes_{j=1}^m u_j^*$ is the tensor that
\begin{equation}\label{Eq:vectorproduct}
     (\otimes_{i=1}^mu_i\otimes_{j=1}^m u_j^*)_{i_1\cdots i_m j_1 \cdots j_m}:= (u_1)_{i_1}\cdots (u_m)_{i_m} (u_1)_{j_1}^*\cdots (u_m)_{j_m}^*.
\end{equation}

\begin{lemma}
\label{Th:proptenstimesvect}
For a Hermitian tensor $\mathcal{A}\in \bH {[n_1, \ldots, n_m]}$
and vectors $u_k\in \mathbb{C}^{n_k}$,
$k\in [m]$, we have:
\bit
\item [i)]
$\mathcal{A}_{i_1\cdots i_m i_1 \cdots i_m}$
and\ $ \langle \mathcal{A}, \otimes_{i=1}^mu_i\otimes_{j=1}^m u_j^* \rangle$
are real;

\item [ii)]
$\mathrm{Tr}_M (\otimes_{i=1}^mu_i\otimes_{j=1}^m u_j^*)
= \| u_1 \|^2 \cdots \| u_m \|^2.$

\eit
\end{lemma}
\begin{proof}
i) Since $\mathcal{A}$ is Hermitian,
$\mathcal{A}^*_{i_1\cdots i_m j_1\cdots j_m}
= \mathcal{A}_{j_1\cdots j_m i_1\cdots i_m }$. By (\ref{Eq:tensorinnerproduct}) and (\ref{Eq:vectorproduct}), we have that
\begin{eqnarray*}
 && \langle \mathcal{A}, \otimes_{i=1}^mu_i\otimes_{j=1}^m u_j^*\rangle \\
 && = \sum_{i_1,\cdots,i_m,j_1,\cdots,j_m=1}^{n_1,\cdots,n_m,n_1,\cdots,n_m} \mathcal{A}^*_{i_1 \cdots i_m j_1 \cdots j_m} (u_1)_{i_1}\cdots (u_m)_{i_m} (u_1)_{j_1}^*\cdots (u_m)_{j_m}^* \\
 && = \sum_{i_1,\cdots,i_m,j_1,\cdots,j_m=1}^{n_1,\cdots,n_m,n_1,\cdots,n_m} \mathcal{A}_{j_1 \cdots j_m i_1 \cdots i_m } (u_1)_{j_1}^*\cdots (u_m)_{j_m}^* (u_1)_{i_1}\cdots (u_m)_{i_m}\\
 && = \sum_{i_1,\cdots,i_m,j_1,\cdots,j_m=1}^{n_1,\cdots,n_m,n_1,\cdots,n_m} \mathcal{A}_{i_1 \cdots i_m j_1 \cdots j_m} (u_1)_{i_1}^*\cdots (u_m)_{i_m}^* (u_1)_{j_1}\cdots (u_m)_{j_m}\\
 && = \langle \mathcal{A}, \otimes_{i=1}^mu_i\otimes_{j=1}^m u_j^*\rangle^* .
\end{eqnarray*}
Hence, $\langle \mathcal{A}, \otimes_{i=1}^mu_i\otimes_{j=1}^m u_j^*\rangle $ is real, the first result follows.

ii) One can check that
$$
\mathrm{Tr}_M (\otimes_{i=1}^mu_i\otimes_{j=1}^m u_j^*)
=\sum_{i_1,\cdots,i_m=1}^{n_1,\cdots,n_m}
(u_1)_{i_1}\cdots (u_m)_{i_m} (u_1)_{i_1}^*\cdots (u_m)_{i_m}^*
$$
$$
= \left( \sum_{i_1=1}^{n_1}(u_1)_{i_1}(u_1)_{i_1}^*
\right) \cdots \left( \sum_{i_m=1}^{n_m}(u_m)_{i_m}(u_m)_{i_m}^*
\right) = \prod_{i=1}^m|| u_i||^2.
$$
This completes the proof.
\end{proof}

Let $Q\in \mathbb{C}^{n_k\times n_k}$ be a square matrix, $k=1,\cdots, m$. The mode-$k$ product 
of a tensor $\mathcal{A}$ by a matrix $Q$ is a $2m$th-order tensor, its entries are given by
\begin{equation}\label{Eq:tensormatrixproduct}
     (\mathcal{A}\times_k Q)_{i_1\cdots i_k \cdots i_{2m}}:= \sum_{t=1}^{n_i} \mathcal{A}_{i_1 \cdots  i_{k-1}t i_{k+1} \cdots i_{2m}} Q_{t i_k},
\end{equation}

\begin{definition}
\rm
Let $\mathcal{A} \in \bH {[n_1, \ldots, n_m]}$ be a Hermitian tensor.
The transformation $\mathcal{A }\rightarrow \mathcal{B}=\mathcal{A}\times_1 Q_1 \cdots \times_m Q_m \times_{m+1} Q_1^* \cdots \times_{2m} Q_m^*$
is called a {\it unitary transformation} if each $Q_k\in \mathbb{C}^{n_k\times n_k}$
is unitary. For such case,
$\mathcal{B}$ is said to be {\it unitary similar} to $\mathcal{A}$.
If all $Q_k$ may be taken to be real (and hence is real orthogonal), then the transformation is said to be (real) orthogonally transformation, $\mathcal{B}$ is said to be (real) orthogonally similar to $\mathcal{A}$.
\end{definition}
Unitary transformation is the extension of unitarily similarity of matrix. As we know, unitarily similar matrices share some common property such as eigenvalues and orthogonality. There are also some invariant properties under unitary transformation stated in (\ref{prop:pro of unitary}).

\begin{proposition}
\label{prop:pro of unitary}
Assume that $\mathcal{A}\in \bH {[n_1, \ldots, n_m]}$
be a Hermitian tensor, $Q_k\in \mathbb{C}^{n_k\times n_k}$
be unitary matrices for $k=1,\cdots, m$.
Let $\mathcal{B}=\mathcal{A}\times_1 Q_1 \cdots \times_m Q_m \times_{m+1} Q_1^* \cdots \times_{2m} Q_m^* $. Then

(i) $\mathcal{B}$ is also a Hermitian tensor;

(ii) $\mathrm{Tr}_M \mathcal{A}= \mathrm{Tr}_M \mathcal{B}$;

(iii) $||\mathcal{A}||_F=||\mathcal{B}||_F$.
\end{proposition}
\begin{proof}
(i) Since $\mathcal{A}$ is a Hermitian tensor, then $ \mathcal{A}^*_{ k_1\cdots k_m t_1\cdots t_m} = \mathcal{A}_{t_1\cdots t_m k_1\cdots k_m}$.
\begin{eqnarray*}
 &&  B^*_{j_1\cdots j_m i_1\cdots i_m }\\
 &&  = \sum_{t_1,\cdots, t_m, k_1,\cdots, k_m=1}^{n_1,\cdots, n_m, n_1,\cdots, n_m} A^*_{ k_1\cdots k_m t_1\cdots t_m} (Q_1)^*_{k_1 j_1}\cdots (Q_m)^*_{k_m j_m} (Q_1)_{t_1 i_1}\cdots (Q_m)_{t_m i_m} \\
 &&  = \sum_{t_1,\cdots, t_m, k_1,\cdots, k_m=1}^{n_1,\cdots, n_m, n_1,\cdots, n_m} A_{t_1\cdots t_m k_1\cdots k_m} (Q_1)_{t_1 i_1}\cdots (Q_m)_{t_m i_m} (Q_1)^*_{k_1 j_1}\cdots (Q_m)^*_{k_m j_m} \\
 &&  = B_{i_1\cdots i_m j_1\cdots j_m}
\end{eqnarray*}
Hence, $\mathcal{B}$ is also a Hermitian tensor.

(ii) Assume that $Q\in \mathbb{C}^{n\times n}$ is a unitary matrix, then
$$
\sum_{i=1}^n (Q)_{t i} (Q)^*_{k i}=\left\{
                                     \begin{array}{ll}
                                       1, & \hbox{if } t=k; \\
                                       0, & \hbox{others. }
                                     \end{array}
                                  \right.
$$
\begin{eqnarray*}
&&  \mathrm{Tr}_M \mathcal{B} = \sum_{i_1, \cdots, i_m} \mathcal{B}_{i_1\cdots i_m i_1\cdots i_m }  \\
   && =\sum_{i_1, \cdots, i_m} \sum_{t_1,\cdots, t_m, k_1,\cdots, k_m} \mathcal{A}_{t_1\cdots t_m k_1\cdots k_m} (Q_1)_{t_1 i_1}\cdots (Q_m)_{t_m i_m} (Q_1)^*_{k_1 i_1}\cdots (Q_m)^*_{k_m i_m}\\
   && =\sum_{t_1,\cdots, t_m, k_1,\cdots, k_m} \mathcal{A}_{t_1\cdots t_m k_1\cdots k_m} 
   \sum_{i_1, \cdots, i_m} (Q_1)_{t_1 i_1}\cdots (Q_m)_{t_m i_m} (Q_1)^*_{k_1 i_1}\cdots (Q_m)^*_{k_m i_m}\\
   && =\sum_{t_1,\cdots, t_m, k_1,\cdots, k_m} \mathcal{A}_{t_1\cdots t_m k_1\cdots k_m} \sum_{i_1} (Q_1)_{t_1 i_1}(Q_1)^*_{k_1 i_1} \cdots \sum_{i_m}(Q_m)_{t_m i_m} (Q_m)^*_{k_m i_m}\\
   && =\sum_{t_1,\cdots, t_m} \mathcal{A}_{t_1\cdots t_m t_1\cdots t_m} = \mathrm{Tr}_M \mathcal{A}
\end{eqnarray*}

(iii)
\begin{eqnarray*}
  ||\mathcal{B}||^2_F &=& \sum_{i_1, \cdots, i_m,j_1,\cdots,j_m} \mathcal{B}_{i_1\cdots i_m j_1\cdots j_m } \mathcal{B}^*_{i_1\cdots i_m j_1\cdots j_m } \\
   &=& \sum_{t_1,\cdots, t_m, k_1,\cdots, k_m} \sum_{t^\prime_1,\cdots, t^\prime_m, k^\prime_1,\cdots, k^\prime_m} \mathcal{A}_{t_1\cdots t_m k_1\cdots k_m} \mathcal{A}^*_{t^\prime_1\cdots t^\prime_m k^\prime_1\cdots k^\prime_m}\\
  & & \cdot\sum_{i_1, \cdots, i_m} (Q_1)_{t_1 i_1}(Q_1)^*_{t^\prime_1 i_1}\cdots (Q_m)_{t_m i_m}(Q_m)^*_{t^\prime_m i_m}\\
 & &\cdot \sum_{j_1,\cdots,j_m}(Q_1)^*_{k_1 j_1}(Q_1)_{k^\prime_1 j_1}\cdots (Q_m)^*_{k_m j_m} (Q_m)_{k^\prime_m j_m}\\
   &=& \sum_{t_1,\cdots, t_m, k_1,\cdots, k_m} \sum_{t^\prime_1,\cdots, t^\prime_m, k^\prime_1,\cdots, k^\prime_m} \mathcal{A}_{t_1\cdots t_m k_1\cdots k_m} \mathcal{A}^*_{t^\prime_1\cdots t^\prime_m k^\prime_1\cdots k^\prime_m}
\end{eqnarray*}
\begin{eqnarray*}
  & & \cdot\sum_{i_1} (Q_1)_{t_1 i_1}(Q_1)^*_{t^\prime_1 i_1}\cdots \sum_{i_m}(Q_m)_{t_m i_m}(Q_m)^*_{t^\prime_m i_m}\\
 & &\cdot \sum_{j_1}(Q_1)^*_{k_1 j_1}(Q_1)_{k^\prime_1 j_1}\cdots \sum_{j_m}(Q_m)^*_{k_m j_m} (Q_m)_{k^\prime_m j_m}\\
 &=& \sum_{t_1,\cdots, t_m, k_1,\cdots, k_m}  \mathcal{A}_{t_1\cdots t_m k_1\cdots k_m} \mathcal{A}^*_{t_1\cdots t_m k_1\cdots k_m} = ||A||^2_F
\end{eqnarray*}

This completes the proof.
\end{proof}

Note: (1) An unitary transformation is a map of Hermitian tensors. However, two unitary similar Hermitian tensors can also be seen as the different representation of the same mixed state under different orthonormal bases.

(2) The matrix trace and the Frobinius norm are invariants of mixed states and Hermitian tensors under unitary transformation.

\section{Partial traces of Hermitian tensors}\label{sc:PartTT}
The concept of the partial trace is first proposed in the quantum mixed state \cite{hjw93} and it takes important role in the quantum information research. Following the same name, we define a partial trace of a Hermitian tensor and study its properties. Moreover, we use partial traces to investigate a sufficient and necessary condition for a complex tensor to be a rank-one tensor.

\begin{definition}\label{def:partialtrace}\rm
Let $\mathcal{H}\in \mathbb{H}[{n_1, \cdots, n_m}]$ be a Hermitian tensor. For $k\in \{1, \cdots, m\}$, define {\it the non-$k$ partial trace} of $\cH$, denoted by $ \mathrm{Tr}_{\bar{k}}(\cH),$
which is a $n_k\times n_k$ matrix with its entries
$$ (\mathrm{Tr}_{\bar{k}}(\cH))_{ij} = \sum_{i_1, \cdots, i_{k-1}, \ i_{k+1}, \cdots, i_m=1}^{n_1, \cdots, n_{k-1},\ n_{k+1}, \cdots, n_m} \mathcal{H}_{{i_1\cdots i_{k-1}\ i \ i_{k+1} \cdots i_m} {i_1\cdots i_{k-1}\ j \ i_{k+1} \cdots i_m}}. $$
Generally, for $I=\{k_1, \cdots, k_s\} $ with $1\leq k_1< \cdots< k_s\leq m$, {\it the non-$I$ partial trace} $ \mathrm{Tr}_{\bar{I}}(\cH)$ of $\cH$ is defined as a $(2s)$th-order Hermitian tensor $  \mathrm{Tr}_{\bar{I}}(\cH)\in \mathbb{H}[n_{k_1}, \cdots, n_{k_s} ]$. Its entries are defined as follows
$$
(\mathrm{Tr}_{\bar{I}}(\cH))_{i_{k_1}\cdots i_{k_s} j_{k_1}\cdots j_{k_s}}=\sum^{n_k}_{i_k=j_k=1,  k\in [m], k \not\in I } \mathcal{H}_{{i_1 i_2\cdots \cdots i_m} {j_1 j_2 \cdots j_m}}.
$$

\end{definition}

Let $\mathcal{A}\in \mathbb{C}^{n_1\times\cdots \times n_m}$ be a complex tensor. Let $\mathcal{A}^*$ be a conjugate tensor of $\mathcal{A}$. A {\it Hermitianlized tensor} of $\mathcal{A}$ is defined as $\rho(\mathcal{A}):=\mathcal{A}\otimes \mathcal{A}^*\in \bH {[n_1, \ldots,  n_m]}$ with its entries as
$$\rho(\mathcal{A})_{i_1\cdots i_m j_1\cdots j_m}:=\mathcal{A}_{i_1\cdots i_m} \mathcal{A}^*_{j_1\cdots j_m}.$$

Let $\rho=\rho(\mathcal{A})$. For each $k\in [m]$,  $ \mathrm{Tr}_{\bar{k}}(\rho)$ is an $n_k\times n_k$ matrix with its entries
$$ (\mathrm{Tr}_{\bar{k}}(\rho))_{ij} = \sum_{i_1, \cdots, i_{k-1}, \ i_{k+1}, \cdots, i_m=1}^{n_1, \cdots, n_{k-1},\ n_{k+1}, \cdots, n_m} \mathcal{A}_{i_1\cdots i_{k-1}\ i \ i_{k+1} \cdots i_m} \mathcal{A}^*_{i_1\cdots i_{k-1}\ j \ i_{k+1} \cdots i_m}. $$

Generally, for $I=\{k_1, \cdots, k_s\} $ with $1\leq k_1< \cdots< k_s\leq m$, {\it the non-$I$ partial trace} $ \mathrm{Tr}_{\bar{I}}(\rho)$ of $\rho$ is a $2s$th-order Hermitian tensor $  \mathrm{Tr}_{\bar{I}}(\rho)\in \mathbb{H}[n_{k_1}, \cdots, n_{k_s} ]$. By definition \ref{def:partialtrace}, its entries are as follows
$$
(\mathrm{Tr}_{\bar{I}}(\rho))_{i_{k_1}\cdots i_{k_s} j_{k_1}\cdots j_{k_s}}=\sum_{i_k=j_k \mathrm{\ for\ } k \not\in I } \mathcal{A}_{i_1 i_2\cdots \cdots i_m} \mathcal{A}^*_{j_1 j_2 \cdots j_m}.
$$

The following example is another way to understand the concept non-$k$ partial trace.
\begin{example}
	Let $e_i=(\underbrace{0, \cdots, 0}_{i-1}, 1, 0, \cdots, 0)^\top\in \mathbb{C}^n$ and $f_j=(\underbrace{0, \cdots, 0}_{j-1}, 1, 0, \cdots, 0)^\top\in \mathbb{C}^m$. A $2$nd-order tensor $\mathcal{A}$ is defined as
	$$
	\mathcal{A}=(a_{ij})_{n\times m}=\sum_{i,j=1}^{n, m} a_{ij} e_i \otimes f_j.
	$$
	\begin{eqnarray}
	\nonumber  \rho(\mathcal{A}) &=& \left(\sum_{i,j=1}^{n, m} a_{ij} e_i \otimes f_j\right)\otimes \left(\sum_{k,l=1}^{n, m} a_{kl} e_k \otimes f_l\right)^* \\
	\label{Eq:4odnontr}   &=& \sum_{i,j=1}^{n, m}\sum_{k,l=1}^{n, m} a_{ij}a^*_{kl} e_i \otimes f_j\otimes e_k \otimes f_l.
	\end{eqnarray}
	
	When $j=l$, the non-$1$ partial trace of $\rho$ is followed by (\ref{Eq:4odnontr}) as
	$$
	Tr_{\bar{1}}(\rho)= \sum_{i,k=1}^{n, n}\left(\sum_{j=1}^{ m} a_{ij}a^*_{kj}\right) e_i \otimes e_k=\left(\sum_{j=1}^{ m} a_{ij}a^*_{kj}\right)_{n\times n}.
	$$
	
\end{example}

The following is the Schmidt polar form, more detail seeing \cite{hjw93}.

\begin{theorem} {\bf(Schmidt polar form)}
	Let $ \mathcal{A} \in \mathbb{C}^{n_1\times n_2}$ be a $2$nd-order tensor. Let $\rho=\rho (\mathcal{A})$, and let $\rho_1=\mathrm{Tr}_{\bar{1}}(\rho)$ and $\rho_2=\mathrm{Tr}_{\bar{2}} (\rho)$ be the partial traces.  Then: %
	
	(1) $\rho_1$ and $\rho_2$ have the same nonzero eigenvalues $\lambda_1, \cdots, \lambda_r$ (with the same multiplicities) and any extra dimensions are made up with zero eigenvalues (noting then $r\leq \min (n_1, n_2)$);
	
	(2) $\mathcal{A}$ can be written as $\mathcal{A}=\sum_{i=1}^r \sqrt{\lambda_i} e_i f_i $, where $e_i $ (respectively $f_i $) are orthonormal eigenvectors of $\rho_1$ in $\mathbb{}C^{n_1}$ (respectively $\rho_2$ in $\mathbb{C}^{n_2}$) belonging to $\lambda_i$. This expression is called the Schmidt polar form of $\mathcal{A}$.
\end{theorem}

However, if $ \mathcal{A}$ is a higher order tensor then the Schmidt polar form does not hold as the following example.

\begin{example}
	We consider a $3$rd-order tensor here. Let $$v(\alpha, \beta)=( \cos \alpha \sin \beta, \sin\alpha \sin \beta, \cos\beta)^T.$$ Let %
	$\mathcal{A}= 0.371391 v(\frac{\pi}{3},\frac{\pi}{3})\otimes v(\frac{\pi}{3}, \frac{5\pi}{6})\otimes v(\frac{-\pi}{6}, \frac{5\pi}{6})%
	+ 0.742781 v(\frac{\pi}{3}, \frac{5\pi}{6})\otimes v(\frac{\pi}{3}, \frac{\pi}{2})\otimes v(\frac{\pi}{3},\frac{\pi}{3})%
	+ 0.557086 v(\frac{\pi}{3},\frac{\pi}{3})\otimes v(\frac{-\pi}{6}, \frac{\pi}{2})\otimes v(\frac{\pi}{3}, \frac{5\pi}{6}),$
	 $\rho= \rho(\mathcal{A})$, $\rho_1=\mathrm{Tr}_{\bar{1}}(\rho)$, $\rho_2=\mathrm{Tr}_{\bar{2}}(\rho)$, and $\rho_3=\mathrm{Tr}_{\bar{3}}(\rho)$. We calculate their eigenvalues directly  \\
	\begin{center}\begin{tabular}{|c|l|}
			\hline
			partial traces & eigenvalues \\ \hline
			$\rho_1$ & $0.57901,\ \ 0.42099,\ \ \ \ 0\ \ $ \\
			$\rho_2$ & $0.624058,\ 0.339349,\ 0.0365928$ \\
			$\rho_3$ & $0.590626,\ 0.383293,\ 0.0260811$ \\
			\hline
	\end{tabular}\end{center}
	It is observed that $\rho_1$, $\rho_2$ and $\rho_3$ have different eigenvalues. Hence, the Schmidt polar form is false for the $3$rd-order tensor.   $\Box$
\end{example}

A decomposition $ \mathcal{A} = \sum_{i=1}^r \lambda_i u_i^{(1)}\otimes \cdots\otimes u_i^{(m)} $ with $0\not=\lambda_i\in \mathbb{R}$ is called orthogonal if $ u_1^{(k)}, \cdots, u_r^{(k)} \in \mathbb{C}^{n_k}$ are normalized and orthogonal for all $k=1, \cdots, m$.

\begin{theorem}\label{Th:main1}
	Assume that $ \mathcal{A} \in \mathbb{C}^{n_1\times\cdots\times n_m}$ is an $m$th-order tensor with an orthogonal decomposition
	$$
	\mathcal{A} = \sum_{i=1}^r \lambda_i u_i^{(1)}\otimes \cdots\otimes u_i^{(m)}.
	$$
	Let $\rho=\rho(\mathcal{A})$, $\rho_k=\mathrm{Tr}_{\bar{k}} (\rho)$. Then $\rho_1, \cdots, \rho_m$ have the same nonzero eigenvalues $\lambda_1^2, \cdots, \lambda_r^2$ (with the same multiplicities), and
	$$
	\rho_k=\sum_{i=1}^r \lambda_i^2 u_i^{(k)}\otimes  u_i^{(k)*}.
	$$
\end{theorem}

\begin{lemma}\label{Lm:lm1}
	Assume that $u_i, v_i \in \mathbb{C}^{n_i}$, $i=1, \cdots, m$. Let $\mathcal{U}=u_1\otimes  \cdots\otimes u_m $, $\mathcal{V}=v_1\otimes  \cdots\otimes v_m $. Then
	$$
	\mathrm{Tr}_M (\mathcal{U}\otimes \mathcal{V}^*)=\mathrm{Tr} (u_1  v_1^*) \cdots \mathrm{Tr} (u_m  v_m^*).
	$$
\end{lemma}
{\bf Proof.} Let $u_i= (u_{i1}, \cdots, u_{i n_i})^T$ and $v_i= (v_{i1}, \cdots, v_{i n_i})^T$,  $i=1, \cdots, m$.
Then
\begin{eqnarray*}
	\mathrm{Tr}_M (\mathcal{U}\otimes \mathcal{V}^*) &=& \sum_{k_1, \cdots, k_m=1}^{n_1, \cdots, n_m} u_{1k_1}\cdots u_{mk_m}v^*_{1k_1}\cdots v^*_{mk_m}  \\
	&=& \left(\sum_{k_1=1}^{n_1} u_{1k_1} v^*_{1k_1}\right) \cdots \left( \sum_{k_m=1}^{n_m}  u_{mk_m} v^*_{mk_m}  \right) \\
	&=& \mathrm{Tr} (u_1 v^*_1) \cdots \mathrm{Tr} (u_m v^*_m).
\end{eqnarray*}
This completes the proof.  $\Box$

{\bf Proof of Theorem \ref{Th:main1}.}  Since $\{ u_1^{(k)}, \cdots, u_r^{(k)}  \}\subset \mathbb{C}^{n_k}$ are normalized and orthogonal for $k=1, \cdots, m$, then
\begin{equation}\label{Eq:ptm1}
\mathrm{Tr}(u_i^{(k)}\otimes u_j^{(k)*})=\delta(i,j)=\left\{                                                                                \begin{array}{ll}
0, & \hbox{if } i\not=j; \\
1, & \hbox{else.}
\end{array} \right.
\end{equation}
By Lemma \ref{Lm:lm1} and (\ref{Eq:ptm1}), it is followed that
\begin{eqnarray*}
	\rho_k &=& \mathrm{Tr}_{\bar{k}} (\rho)= \sum_{i,j=1}^r \lambda_i \lambda_j  \mathrm{Tr}_{\bar{k}} \left(u_i^{(1)}\otimes  u_j^{(1)*}\otimes \cdots\otimes u_i^{(m)}\otimes  u_j^{(m)*}\right) \\
	&=& \sum_{i,j=1}^r \lambda_i \lambda_j\ \mathrm{Tr}_M\left(\prod_{t=1, t\not=k}^m u_i^{(t)}\otimes  u_j^{(t)*} \right) u_i^{(k)}\otimes  u_j^{(k)*} \\
	&=& \sum_{i,j=1}^r \lambda_i \lambda_j\ \left(\prod_{t=1, t\not=k}^m \mathrm{Tr}(u_i^{(t)}\otimes  u_j^{(t)*}) \right) u_i^{(k)}\otimes  u_j^{(k)*}\\
	&=& \sum_{i=1}^r \lambda_i^2 u_i^{(k)}\otimes u_i^{(k)*}
\end{eqnarray*}
This completes the proof.   $\Box$

\begin{theorem}\label{Th:rankone1}
	Let $ \mathcal{A} \in \mathbb{C}^{n_1\times\cdots\times n_m}$ be an normalized $m$th-order tensor.
	Let $\rho=\rho(\mathcal{A})$, $\rho_k=\mathrm{Tr}_{\bar{k}} (\rho)$. Then $\mathcal{A}$ is a rank-one tensor iff $\rho_1, \cdots, \rho_m$ have the same only one nonzero eigenvalue $\lambda=1$.
\end{theorem}
{\bf Proof.} The necessity is followed by Theorem \ref{Th:main1} directly for $r=1$. Now we prove the sufficiency.  Since $\rho_1, \cdots, \rho_m$ have the same only one nonzero eigenvalue $\lambda=1$, then $\rho_1, \cdots, \rho_m$ are rank-one matrices and can be written as $\rho_k=  u^{(k)} u^{(k)*}$, $k=1, \cdots, m$.

Firstly, since $\rho_1$ has a singer nonzero eigenvalue $\lambda=1$, we see $\mathcal{A}$ as a $n_1\times (n_2\times\cdots\times n_m)$ matrix, by Schmidt polar form, there is a normalized tensor $\tilde{\mathcal{U}}^{(1)} \in \mathbb{C}^{n_2\times\cdots\times n_m}$, such that
$$ \mathcal{A} =  u^{(1)}\otimes \tilde{\mathcal{U}}^{(1)} . $$

It follows that
$$\rho_2= \mathrm{Tr}_{\bar{2}}(\rho)=  \mathrm{Tr}_{\bar{2}}(u^{(1)}\otimes \tilde{\mathcal{U}}^{(1)}\otimes u^{(1)*}\otimes \tilde{\mathcal{U}}^{(1)*})=  \mathrm{Tr}_{\bar{1}} ( \tilde{\mathcal{U}}^{(1)}\otimes \tilde{\mathcal{U}}^{(1)*}). $$
Secondly, since $\rho_2$ has a singer nonzero eigenvalue $\lambda=1$, we see $\tilde{\mathcal{U}}^{(1)}$ as a $n_2\times (n_3\times\cdots\times n_m)$ matrix, again by Schmidt polar form, there is a normalized state $\tilde{\mathcal{U}}_2  \in \mathbb{C}^{n_3\times\cdots\times n_m}$, such that
$$ \tilde{\mathcal{U}}^{(1)}  =  u^{(2)}\otimes \tilde{\mathcal{U}}^{(2)},\ \mathrm{and}\
\rho_3=  \mathrm{Tr}_{\bar{1}} ( \tilde{\mathcal{U}}^{(2)}\otimes \tilde{\mathcal{U}}^{(2)*}). $$
And so on, it follows that
$$\mathcal{A} =  u^{(1)} \cdots u^{(m)}. $$
This completes the proof. $\Box$

\begin{corollary}\label{Th:rankone2}
	Let $ \mathcal{A} \in \mathbb{C}^{n_1\times\cdots\times n_m}$ be a normalized $m$th-order tensor.
	Let $\rho=\rho(\mathcal{A})$, $\rho_k=\mathrm{Tr}_{\bar{k}} (\rho)$. Then $\mathcal{A}$ is a rank-one tensor iff $\mathrm{Det}(\rho_k-I_k)=0$, where $I_k$ is a unit matrix, $k=1, \cdots, m$.
\end{corollary}
{\bf Proof.} Since $\mathcal{A}$ is normalized and all nonzero eigenvalues of $\rho_k$ are positive, hence the sum of all eigenvalues (including multiplicities) of $\rho_k$ equals 1. It implies that if 1 is a nonzero eigenvalue of $\rho_k$ then 1 is the single nonzero eigenvalue of $\rho_k$. By Theorem \ref{Th:rankone1}, it follows that $\mathcal{A}$ is a rank-one tensor iff 1 is a nonzero eigenvalue of $\rho_k$. By matrix theory, we know that 1 is an eigenvalue of $\rho_k$ iff $\mathrm{Det}(\rho_k-I_k)=0$. Hence, the result hold.  $\Box$

\begin{theorem}
	Let $ \mathcal{A}, \mathcal{B} \in \mathbb{C}^{n_1\times\cdots\times n_m}$ be $m$th-order tensors. Let $\rho^A=\rho(\mathcal{A})$, $\rho_k^A=\mathrm{Tr}_{\bar{k}} (\rho^A)$, $\rho^B=\rho(\mathcal{B})$, $\rho_k^B=\mathrm{Tr}_{\bar{k}} (\rho^B)$. Then $\rho_k^A$ and $\rho_k^B$ are unitary similar for $k=1, \cdots, m$ if $\mathcal{A}$ and $\mathcal{B}$ are unitary similar.
\end{theorem}
{\bf Proof.} Assume that $\{e^{(k)}_{1}, \cdots, e^{(k)}_{n_k}\}$ is an orthonormal  basis of $\mathbb{C}^{n_k}$, $k=1, \cdots, m$, and
$$\rho^\mathcal{A}=\sum A_{i_1\cdots i_m}A^*_{j_1\cdots j_m} e^{(1)}_{i_1}\otimes \cdots \otimes e^{(m)}_{i_m}\otimes e^{(1)*}_{j_1}\otimes \cdots \otimes e^{(m)*}_{j_m}. $$
Then
$$(\rho_k^\mathcal{A})_{ij} = (\sum A_{i_1\cdots i_{k-1} i i_{k+1}\cdots i_m}A^*_{i_1\cdots i_{k-1} j i_{k+1}\cdots i_m})  e^{(k)}_{i} \otimes  e^{(k)*}_{j}.$$

Since $\mathcal{A}$ and $\mathcal{B}$ are unitary equivalent, hence there are unitary matrices $ Q_1, \cdots, Q_m$ such that $\rho^B=\rho^A\times_1 Q_1\cdots\times_m Q_m\times_{m+1} Q_1^* \cdots\times_{2m} Q_m^*$.
Let $f^{(k)}_{i}= Q_k e^{(k)}_{i}$, $i=1, \cdots, n_k$. Then $\{f^{(k)}_{1}, \cdots, f^{(k)}_{n_k}\}$ is another orthonormal basis of $\mathbb{C}^{n_k}$, and
$$ \rho^\mathcal{B}=\sum A_{i_1\cdots i_m}A^*_{j_1\cdots j_m} f^{(1)}_{i_1}\otimes \cdots \otimes f^{(m)}_{i_m}\otimes  f^{(1)*}_{j_1}\otimes \cdots \otimes f^{(m)*}_{j_m}. $$
Hence,
$$(\rho_k^\mathcal{B})_{ij} = (\sum A_{i_1\cdots i_{k-1} i i_{k+1}\cdots i_m}A^*_{i_1\cdots i_{k-1} j i_{k+1}\cdots i_m})  f^{(k)}_{i}\otimes   f^{(k)*}_{j}.$$

It follows that
$$ \rho_k^\mathcal{B}= \rho_k^\mathcal{B}\times_1 Q_k \times_2 Q^*_k. $$
Hence, $\rho_k^A$ and $\rho_k^B$ are unitary equivalent. This completes the proof.  $\Box$

This means that eigenvalues of partial traces are unchanged under unitary transformations.


\section{Nonnegativity and Hermitian Eigenvalues}
\label{sc:nhe}

For a Hermitian tensor $\cH \in \bH {[n_1, \ldots, n_m]}$,
recall that $\cH(x)$ is the conjugate polynomial
\[
\cH(x) = \langle \cH,  \otimes_{i=1}^m x_i \otimes_{j=1}^m  x_j^* \rangle,
\]
in $x := (x_1, \ldots, x_m)$, with complex variables
$x_1 \in \cpx^{n_1}, \ldots, x_m \in \cpx^{n_m}$.
Note that $\cH(x)$ always achieves real values and
$\cH(x)$ is Hermitian quadratic in each $x_i$.

\begin{definition}
\rm \label{def:nng}
A Hermitian tensor $\cH$ is called {\it nonnegative}
(resp., {\it positive}) if $\cH(x) \geq 0$ (resp., $\cH(x) > 0$)
for all $\| x_1 \|= \cdots = \| x_m \| =1$.
The set of all nonnegative Hermitian tensors is denoted as
\[
{\rm NN\mathbb{H}} {[n_1,\ldots,n_m]} :=
\left\{\cH \in \bH {[n_1, \ldots, n_m]}: \,
\cH(x) \geq 0, \, \forall \, x_i \in \cpx^{n_i}
\right \}.
\]
\end{definition}

\begin{proposition}
The set ${\rm NN\mathbb{H}} {[n_1,\ldots,n_m]}$ is a proper cone,
i.e., it is closed, convex, pointed and solid.
\end{proposition}
\begin{proof}
Clearly, ${\rm NN\mathbb{H}} {[n_1,\ldots,n_m]}$ is a closed, convex cone.
It is solid, i.e., it has an interior point.
For instance, the tensor $\mc{I}$ such that
\[
\mc{I}(x) = (x_1^*x_1)\cdots (x_m^*x_m)
\]
is an interior point, because the minimum value of $\mc{I}(x)$
over the spheres $\|x_1\| = \cdots = \| x_m \|=1$ is one.
The cone ${\rm NN\mathbb{H}} {[n_1,\ldots,n_m]}$ is also pointed.
This is because if $\cH \in {\rm NN\mathbb{H}} {[n_1,\ldots,n_m]}$ and
$-\cH  \in {\rm NN\mathbb{H}} {[n_1,\ldots,n_m]}$, then $\cH =0$.
This is because $\cH(x) \equiv 0$ on $\|x_1\| = \cdots = \| x_m \|=1$
and $\cH$ is Hermitian.
\end{proof}

The nonnegativity or positivity of a Hermitian tensor is related to its
Hermitian eigenvalues, which we define as follows.
Consider the optimization problem
\begin{equation}
\label{minH:||x||=1}
\min  \quad   \cH(x)  \quad
\mbox{s.t.}  \quad  x_1^*x_1 = 1, \ldots,  x_m^*x_m = 1.
\end{equation}
The first order optimality conditions for
\reff{minH:||x||=1} are
\begin{equation}
\nn
  \langle \cH, \otimes_{i=1,i\not=k}^m x_i
   \otimes_{j=1}^m x_j^*\rangle  = \lambda_k x_k^*,
\end{equation}
\begin{equation}
\nn
  \langle \cH, \otimes_{i=1}^m x_i \otimes_{j=1,j\not=k}^m x_j^*  \rangle
  = \lambda_k x_k,
\end{equation}
where $\lmd_k$ is the Lagrange multiplier for
$x_k^*x_k=1$, for $k=1,\ldots, m$. Because of the
constraints $x_k^*x_k=1$,
one can show that all Lagrange multipliers are equal.
So we can write them as
\begin{equation}
\label{Eq:HermitEiegnvalue01}
  \langle \cH, \otimes_{i=1,i\not=k}^m x_i
\otimes_{j=1}^m x_j^*\rangle = \lambda  x_k^*,
\end{equation}
\begin{equation}
\label{Eq:HermitEigenvalue02}
  \langle \cH, \otimes_{i=1}^m x_i \otimes_{j=1,j\not=k}^m x_j^*  \rangle
  = \lambda  x_k.
\end{equation}
Clearly, we can get
$$
\lambda = \langle \cH, \otimes_{i=1}^m x_i \otimes_{j=1}^m x_j^*  \rangle = \cH(x).
$$
Hence, by Theorem \ref{Th:proptenstimesvect},
we know that $\lambda$ must be real. Since,
$$
\langle \cH, \otimes_{i=1}^m x_i \otimes_{j=1,j\not=k}^m x_j^*  \rangle^*
 = \langle \cH, \otimes_{i=1,i\not=k}^m x_i \otimes_{j=1}^m x_j^*  \rangle,
$$
hence equations (\ref{Eq:HermitEiegnvalue01})
and (\ref{Eq:HermitEigenvalue02}) are equivalent.

\begin{definition}
\rm  \label{def:HMeig}
For a Hermitian tensor $\cH \in \bH {[n_1, \ldots, n_m]}$,
if a tuple $( \lambda; u_1, \cdots, u_m)$,
with each $\| u_i \| = 1$, satisfies (\ref{Eq:HermitEiegnvalue01})
or (\ref{Eq:HermitEigenvalue02}) for $k=1,\ldots,m$,
then $\lambda$ is called a {\it Hermitian eigenvalue},
and $(\lambda; u_1, \cdots, u_m)$
is called a {\it Hermitian eigentuple}.
In particular, $u_i$ is called the mode-$i$ {\it Hermitian eigenvector},
and $\otimes_{i=1}^mu_i\otimes_{j=1}^m u_j^*$
is called the {\it Hermitian eigentensor}.
\end{definition}

Clearly, the largest (resp., smallest) Hermitian eigenvalue of $\cH$
is the maximum (resp., minimum) value of
$\cH(x)$ over the multi-spheres $ \| x_i \|=1$.
Consequently,
A Hermitian tensor $\cH$ is nonnegative (resp., positive)
if and only if all its Hermitian eigenvalues are greater than
or equal to zero (resp., strictly bigger than zero).

\section{Hermitian decomposition }
\label{sc:R1HTD}

\begin{definition}
\rm
For a Hermitian tensor
$\mathcal{H}\in \bH {[n_1, \ldots, n_m]}$,
if it can be written as
\begin{equation} \label{Eq:rank-oneHermitDecomp}
\mathcal{H}=\sum_{i=1}^r \lmd_i \,  u_i^{(1)}\otimes \ldots \otimes u_i^{(m)} \otimes u_i^{(1)*}\otimes \ldots \otimes u_i^{(m)*}
\end{equation}
for $\lmd_i\in \mathbb{R}$, $u_i^{(j)}\in \mathbb{C}^{n_j}$
and $\|u_i^{(j)} \|=1$,
then $\mathcal{H}$ is called {\it Hermitian decomposable}.
In this case, \reff{Eq:rank-oneHermitDecomp}
is called a {\it Hermitian decomposition} of $\mathcal{H}$.
The smallest number $r$ in \reff{Eq:rank-oneHermitDecomp}
is called the {\it Hermitian rank} of $\mathcal{H}$,
which we denote as $\rank_H(\cH)$. If all $\lambda_i>0$, then (\ref{Eq:rank-oneHermitDecomp}) is called a {\it positive Hermitian decomposition} of $\mathcal{H}$, and $\cH$ is called {\it positive Hermitian decomposable}.
\end{definition}

It is well known that every tensor and symmetric tensor have their canonical decomposition and symmetric decomposition respectively. So a natural question is that whether every Hermitian tensor is Hermitian decomposable or not. Fortunately the answer is yes.

\begin{theorem}
Every Hermitian tensor
$\mathcal{H} \in \bH {[n_1, \ldots, n_m]}$
is Hermitian decomposable.
\end{theorem}

\begin{proof}

It is clear that for each $\mathcal{A}, \mathcal{B}\in \bH {[n_1, \ldots, n_m]}$ and $a, b\in \mathbb{R}$, then $a \mathcal{A}+b \mathcal{B}\in \bH {[n_1, \ldots, n_m]}$, which means that $\bH {[n_1, \ldots, n_m]}$ is a linear space over $\mathbb{R}$. Denote $\bH {[n_1, \ldots, n_m]}_1$ as the set of all rank-1 Hermitian decomposable tensors in $\bH {[n_1, \ldots, n_m]}$. Then $\bH {[n_1, \ldots, n_m]}_1$ is also a linear space over $\mathbb{R}$, and it is a subspace of $\bH {[n_1, \ldots, n_m]}$. In the following, we will show that $\bH {[n_1, \ldots, n_m]}$ and $\mathbb{H}{[n_1,\cdots,n_m]}_1$ have the same dimension.

Denote $[n_1, \cdots, n_m] :=\{ (i_1, \cdots, i_m)| i_1\in [n_1], \cdots, i_m\in [n_m]  \}$. Denote $\mathcal{E}_{i_1\cdots i_m j_1\cdots j_m}$ as a $2m$th-order tensor with only one nonzero entry $(\mathcal{E}_{i_1\cdots i_m j_1\cdots j_m})_{i_1\cdots i_m j_1\cdots j_m}=1$. Denote $I :=(i_1, \cdots, i_m)$, $J :=(j_1, \cdots, j_m)$. Then $\mathcal{E}_{IJ}=\mathcal{E}_{i_1\cdots i_m j_1\cdots j_m}$. Define an order $ I < J $ if there is a number $k\in [m]$ such that $i_1=j_1, \cdots, i_{k-1}=j_{k-1}$ and $ i_k < j_k$.

On the one hand, let $E_1=\{ \mathcal{E}_{II}: I\in [n_1, \cdots, n_m] \}$,  $E_2= \{ \mathcal{E}_{I J} + \mathcal{E}_{J I } : I<J, I, J\in [n_1, \cdots, n_m] \}$, $ E_3= \{ \sqrt{-1} \mathcal{E}_{IJ} - \sqrt{-1} \mathcal{E}_{JI }:I, J\in [n_1, \cdots, n_m], I<J \}$. Then  $E_1\bigcup E_2\bigcup E_3$ is a basis of the linear space $\mathbb{H} {[n_1,\cdots, n_m]}$ over $\mathbb{R}$.
Since $\#E_1=N$, $\#E_2=\#E_3= \frac{N(N-1)}{2}$, where $\#$ denotes the number of entries of the set. Hence, The dimension of $\bH {[n_1, \ldots, n_m]}$ is $n^2$, where $n=n_1\times\cdots\times n_m$.

On the other hand, let $\{e_{i_k}^{(k)}\otimes e_{i_k}^{(k)*}\}_{i_k=1}^{D_k} $ is a basis of the linear space $\mathbb{H}[n_k]$  over $\mathbb{R}$, $D_k$ is the dimension. Then $D_k=n_k^2$. Let $$E=\{e_{i_1}^{(1)}\otimes \cdots\otimes e_{i_m}^{(m)}\otimes e_{i_1}^{(1)*}\otimes \cdots\otimes e_{i_m}^{(m)*} | i_k=1, \cdots, D_k, k=1, \cdots, m \}.$$  Then $E$ is a basis of the linear space $\bH {[n_1, \ldots, n_m]}_1$ over $\mathbb{R}$,
and its demission $\#E=n_1^2\times\cdots\times n_m^2=n^2$.

Hence, $\bH {[n_1, \ldots, n_m]}_1 = \bH {[n_1, \ldots, n_m]}$.
It follows that every Hermitian tensor is Hermitian decomposable.
\end{proof}

Let $n=n_1\times \cdots\times n_m$. Every Hermitian tensor $\mathcal{H}$
can be flattened as a Hermitian matrix $H \in \cpx^{ n\times n}$,
labeled in the way that
\begin{equation}\label{HIJ=Hij}
(H)_{I,J} = \cH_{i_1 \ldots i_m j_1 \ldots j_m}
\end{equation}
for $I:=(i_1,\ldots, i_m)$ and $J:=(j_1,\ldots, j_m)$.
For a tensor $\mathcal{U} \in \mathbb{C}^{n_1\times\cdots\times n_m}$,
$\mathcal{U}^*$ denotes the tensor obtained
by applying complex conjugates to its entries.
Note that $\mathcal{U}\otimes \mathcal{U}^*$
is always Hermitian, because
\[
(\mathcal{U} \otimes \mathcal{U}^*)_{i_1\cdots i_m j_1\cdots j_m}
=
(\mathcal{U})_{i_1\cdots i_m} (\mathcal{U}^*)_{j_1\cdots j_m}.
\]
The following is the spectral theorem for Hermitian tensors.

\begin{theorem}\label{Th:TensorSpectralTheorem}
For every Hermitian tensor $\mathcal{H}\in \bH {[n_1, \ldots, n_m]}$,
there exist nonzero real numbers
$\lambda_i \in \re$ and tensors
$\mathcal{U}_i\in \mathbb{C}^{n_1\times\cdots\times n_m}$ such that
\begin{equation} \label{Eq:tensormatrixdecom}
   \mathcal{H}=\sum_{i=1}^s \lambda_i \mathcal{U}_i \otimes \mathcal{U}_i^*,
   \quad \mbox{where} \quad
   \langle \mathcal{U}_i, \mathcal{U}_i\rangle =1, \,
   \quad \langle \mathcal{U}_i, \mathcal{U}_j \rangle =0 \, (i\not= j).
\end{equation}
\end{theorem}
\begin{proof}
Let $H$ be the matrix labeled as in \reff{HIJ=Hij}.
By the definition, the tensor $\cH$ is Hermitian if and only if
the matrix $H$ is Hermitian, which is then equivalent to that
\[
H = \sum_{i=1}^s \lmd_i  q_i q_i^*
\]
for real scalars $\lmd_i$ and orthonormal vectors
$q_1, \ldots, q_s \in \cpx^N$.
We label vectors in $\cpx^N$ by $I=(i_1, \ldots, i_m)$.
Each $q_i \in \cpx^N$ can be folded into a tensor
$\mathcal{U}_i \in \cpx^{n_1 \times \cdots \times n_m}$
such that
\[
(q_i)_I = (\mathcal{U}_i)_{i_1 \ldots i_m}.
\]
The above decomposition for $H$ is then equivalent to
\[
\cH = \sum_{i=1}^s \lmd_i \mathcal{U}_i \otimes \mathcal{U}_i^*.
\]
Also note that
\[
\langle \mathcal{U}_i, \mathcal{U}_j \rangle
= \langle q_i, q_j \rangle = q_i^* q_j,
\]
which equals $1$ for $i=j$ and zero otherwise.
\end{proof}

From the proof, we can see that the real scalars
$\lmd_1, \cdots, \lmd_s$ in \reff{Eq:tensormatrixdecom}
are eigenvalues of the Hermitian matrix $H$.
We call them {\it matrix eigenvalues} of $\cH$.
The Hermitian eigenvalues are defined as in
\reff{Eq:HermitEiegnvalue01}-\reff{Eq:HermitEigenvalue02}.
The equation \reff{Eq:tensormatrixdecom} is called an
{\it eigen-matrix decomposition} of $\mathcal{H}$.

%
%

It is well known that there are both separable states and entangled states in mixed quantum states. Hence, even if Hermitian tensors are Hermitian decomposable, but not any Hermitian tensor has a positive Hermitian decomposition. Denote $\mathrm{PHD}{[n_1,\ldots,n_m]}$ as the set of all positive Hermitian decomposable tensors.
Recall the cone of nonnegative Hermitian tensors:
\[
{\rm NN\mathbb{H}} {[n_1,\ldots,n_m]} :=
\left\{\cH \in \bH {[n_1, \ldots, n_m]}: \,
 \cH(x) \geq 0, \
x=(x_1,\ldots, x_m ),\ \forall \, x_i \in \cpx^{n_i}
\right \}.
\]


\begin{theorem}
If $\cH$ is a positive Hermitian decomposable tensor,
then all the matrix eigenvalues and Hermitian eigenvalues
of $\cH$ are nonnegative.
\end{theorem}
\begin{proof}
When $\cH$ is positive Hermitian decomposable,
it has a decomposition as
\begin{equation}\label{Eq:separabletensor}
  \mathcal{H}=\sum_{i=1}^r    u_i^{(1)}\otimes \ldots \otimes u_i^{(m)} \otimes u_i^{(1)*}\otimes \ldots \otimes u_i^{(m)*}
\end{equation}

Its Hermitian flattening matrix $H$ takes the form
\[
H = \sum_{i=1}^r   \mathcal{U}_i \mathcal{U}_i^*,
\]
where each $\mathcal{U}_i$ is the vector corresponding to the tensor
$u_i^{(1)} \otimes \cdots \otimes u_i^{(m)}$.
Clearly, the matrix $H$ is positive semidefinite,
hence all the matrix eigenvalues are nonnegative.
Moreover, we also have
\[
\cH(x)=\langle \cH, \otimes_{i=1}^m x_i \otimes_{j=1}^m x_j^*  \rangle
= \sum_{i=1}^r |u_i^{(1)*} x_i|^2.
\]
It is alway nonnegative over the multi-sphere
$\| x_i \|=1$. So, the critical values of
$\cH(x) $
are all nonnegative, i.e.,
all the Hermitian eigenvalues are nonnegative.
\end{proof}

\begin{theorem}\label{Th:HP-TensMatrixDecompos}
	(Hughston-Jozsa-Wootters, 1993 \cite{hjw93}) Let $n=n_1\times n_2\times\ldots\times n_m$. Assume that $\mathcal{H}$ is a Hermitian tensor with a positive eigen-matrix decomposition
\begin{equation}\label{Eq:PEigenD2}
\cH = \sum_{i=1}^s \mathcal{U}_i\otimes \mathcal{U}_i^*,
\end{equation}
	and  a positive Hermitian decomposition
\begin{equation}\label{Eq:PD2}
\cH = \sum_{i=1}^r \mathcal{V}_i\otimes \mathcal{V}_i^*.
\end{equation}	
	Let $U=(\mathcal{U}_1, \mathcal{U}_2, \ldots, \mathcal{U}_s ) $ be an $n\times s$ matrix, and $V=(\mathcal{V}_1, \mathcal{V}_2, \ldots, \mathcal{V}_r ) $ be an $n\times r$ matrix, respectively. Then $r\geq s$ and there is an $s\times r$ matrix $Q$ satisfying $ Q Q^\dag = I_{s \times s} $, such that $V= U Q$, where $I_{s \times s} $ denotes the $s \times s$ unit matrix. Further more, if $r>s$, then $Q$ can be extended to an $ r\times r $ unitary matrix $P$, such that $ (U, 0) = V P^{-1},$ where $(U, 0)$ is an $n\times r$ matrix.
\end{theorem}

\begin{theorem}\label{Th:HermitianTensorDecompositionC}
	Assume that $\mathcal{H}$ is a Hermitian tensor with a positive decomposition (\ref{Eq:tensormatrixdecom}) and a positive Hermitian decomposition (\ref{Eq:rank-oneHermitDecomp}) with $p_i$ and $\lambda_j$ are positive for $i=1, \cdots, r$,  $j=1, \cdots, s$. Let $x_{ij}=\langle \mathcal{U}_j, u_i^{(1)}\cdots u_i^{(m)} \rangle$, $Q_{ij}=\sqrt{p_i/\lambda_j} x_{ij}$, for $i=1, \cdots, r$,  $j=1, \cdots, s$. Then
	$r\geq s$, $Q^\dag Q=I_{s\times s}$, and
	$$
	\left(   \begin{array}{ccc}
	\sqrt{p_1} u_1^{(1)}\cdots u_1^{(m)}\\ \cdots \\ \sqrt{p_r} u_r^{(1)}\cdots u_r^{(m)}\\ \end{array}  \right)
	= Q  \left(   \begin{array}{ccc}
	\sqrt{\lambda_1}\mathcal{U}_1\\ \cdots \\ \sqrt{\lambda_s}\mathcal{U}_s\\
	\end{array} \right), \
	\left(
	\begin{array}{ccc}
	\sqrt{\lambda_1}\mathcal{U}_1\\ \cdots \\ \sqrt{\lambda_s}\mathcal{U}_s\\
	\end{array}
	\right)
	=Q^\dag \left(
	\begin{array}{ccc}
	\sqrt{p_1} u_1^{(1)}\cdots u_1^{(m)}\\ \cdots \\ \sqrt{p_r} u_r^{(1)}\cdots u_r^{(m)}\\
	\end{array}
	\right).
	$$
	
\end{theorem}

{\bf Proof.} Since $x_{ij}=\langle \mathcal{U}_j, u_i^{(1)}\cdots u_i^{(m)} \rangle$ and $Q_{ij}=\sqrt{p_i/\lambda_j} x_{ij}$ for $i=1, \cdots, r$,  $j=1, \cdots, s$, then $ Q=  (\mathrm{diag} (\frac{1}{|\lambda_1|}, \cdots, \frac{1}{|\lambda_s|}) U^\dag V)^T $. By Theorem \ref{Th:HP-TensMatrixDecompos}, it follows that $r\geq s$, $Q^\dag Q = I_{s\times s} $, $V^T= Q U^T $ and $ U^T  =Q^\dag V^T.$ Hence, these results are followed. $\Box$

From Theorem \ref{Th:HP-TensMatrixDecompos} and Theorem \ref{Th:HermitianTensorDecompositionC}, we know that if $\mathcal{H}$ is a Hermitian tensor with a positive eigen-matrix decomposition (\ref{Eq:PEigenD2}) and a positive Hermitian decomposition (\ref{Eq:PD2}), then Span($\mathcal{U}_1, \cdots, \mathcal{U}_s$)=Span($\mathcal{V}_1, \cdots, \mathcal{V}_r$), and there is a matrix $Q_{s\times r}$ such that $ Q Q^\dag = I_{s\times s} $ and $V=UQ$. Hence, by this method, one can find 
a positive Hermitian decomposition of $\mathcal{H}$, or determine that $\mathcal{H}$ is not positive Hermitian decomposable.

\section{Application to quantum mixed state}\label{sec:mixedstates}
Let $\rho$ be an $m$-partite mixed state.  Let $\{ |e_i^{(k)}\rangle|i=1, \cdots, n_k  \} $ is an orthonormal basis of the $k$-th system for all $k\in [m]$ and $\mathcal{H}\in \mathbb{H}[n_1, \cdots, n_m]$ is the corresponding Hermitian tensor of $\rho$.
Assume that $\{ |f_i^{(k)}\rangle|i=1, \cdots, n_k  \} $ is another orthonormal basis of the $k$-th system for all $k\in [m]$ and $\mathcal{T}$ is the corresponding Hermitian tensor of $\rho$ under the orthonormal basis. Then, $\mathcal{H}$ is unitary similar to $\mathcal{T}$, and the state $\rho$ is separable if and only if $\mathcal{H}$ is positive Hermitian decomposition. From the above sections discussion, we have the following properties of mixed states.

\begin{theorem}\label{Th:mixedstates}
Assume that $\rho$ is a quantum mixed state and $\mathcal{H}\in \mathbb{H}[n_1, \cdots, n_m]$ is the corresponding Hermitian tensor of $\rho$ under an orthonormal basis. The following results are true.

(1) If the smallest matrix eigenvalue is negative, then the state $\rho$ is entangled.

(2) If the smallest Hermitian eigenvalue is negative, then the state $\rho$ is entangled.

(3) Assume that $\mathcal{H}$ has a positive eigen-matrix decomposition (\ref{Eq:PEigenD2}). Let $$K=Span\{ \mathcal{U}_1, \mathcal{U}_2, \cdots, \mathcal{U}_s \}.$$ If $\mathcal{H}$ has a positive Hermitian decomposition (\ref{Eq:rank-oneHermitDecomp}), then $u_i^{(1)}\otimes \ldots \otimes u_i^{(m)}\in K $ for all $i\in [r]$.
	
\end{theorem}

\begin{example}
Let $|\psi_1\rangle = (|00\rangle +|01\rangle +\sqrt{-1} |11\rangle)/\sqrt{3}$, $|\psi_2\rangle = (|00\rangle -|01\rangle +4\sqrt{-1} |10\rangle)/(3\sqrt{2})$, $\rho =\rho_1 |\psi_1\rangle \langle \psi_1|+ \rho_2 |\psi_2\rangle \langle \psi_2| $, where $\rho_1>0, \rho_2>0$ and $\rho_1+\rho_2=1$. Next, let's discuss whether the state $\rho$ is separable or entangled. We will take four steps to deal with the problem.

Step 1: Let
\begin{equation}\label{ex:mixedPHD}
\mathcal{U}_1=\frac{1}{\sqrt{3}}\left(
                \begin{array}{cc}
                  1 & 1 \\
                  0 & \sqrt{-1} \\
                \end{array}
              \right),
\mathcal{U}_2=\frac{1}{3\sqrt{2}}\left(
                \begin{array}{cc}
                  1 & -1 \\
                  4\sqrt{-1} & 0 \\
                \end{array}
              \right),
\mathcal{H}=\rho_1 \mathcal{U}_1\otimes \mathcal{U}_1^*+\rho_2 \mathcal{U}_2\otimes \mathcal{U}_2^*.
\end{equation}
Then $\mathcal{U}_1 $, $\mathcal{U}_2 $ and $\mathcal{H} $ are the corresponding tensors of $|\psi_1\rangle$, $|\psi_2\rangle$ and $\rho$ under the orthonormal basis $\{|0\rangle,|1\rangle \}$, respectively. Science $\langle \mathcal{U}_i,\mathcal{U}_j\rangle = \delta(i,j)$, $i,j=1,2$, hence (\ref{ex:mixedPHD}) is an eigen-matrix decomposition of $\mathcal{H}$.

Step 2: Let
$$
K=\{k_1 \mathcal{U}_1+k_2 \mathcal{U}_2 | k_1,k_2\in \mathbb{C}\}=\left\{\left(
                \begin{array}{cc}
                  \frac{k_1}{\sqrt{3}}+\frac{k_2}{3\sqrt{2}} & \frac{k_1}{\sqrt{3}}-\frac{k_2}{3\sqrt{2}} \\
                  \frac{4k_2\sqrt{-1}}{3\sqrt{2}} & \frac{k_1\sqrt{-1}}{\sqrt{3}} \\
                \end{array}
              \right) | k_1,k_2\in \mathbb{C}\right\}.
$$
Obviously, a matrix in $K$ is a rank-one matrix if and only if $k_2 =\frac{3 \sqrt{6}  - \sqrt{-42} }{8}k_1$ or $k_2 = \frac{3 \sqrt{6}  + \sqrt{-42} }{8}k_1$.

If $k_2 =\frac{3 \sqrt{6}  - \sqrt{-42} }{8}k_1$, then
$$
k_1 \mathcal{U}_1+k_2 \mathcal{U}_2=\frac{k_1}{\sqrt{3}}\left(
                                      \begin{array}{cc}
                                        (11 -  \sqrt{-7})/8  & (5 +  \sqrt{-7})/8  \\
                                        (3 - \sqrt{-7})/2 & 1 \\
                                      \end{array}
                                    \right)
$$
$$
=\frac{k_1}{\sqrt{3}} \left(
   \begin{array}{c}
     (5+\sqrt{-7})/8 \\
     1 \\
   \end{array}
 \right)\otimes
 \left(
   \begin{array}{c}
     (3-\sqrt{-7})/2 \\
     1 \\
   \end{array}
 \right).
$$

If $k_2 =\frac{3 \sqrt{6}  + \sqrt{-42} }{8}k_1$, then
$$
k_1 \mathcal{U}_1+k_2 \mathcal{U}_2=\frac{k_1}{\sqrt{3}}\left(
                                      \begin{array}{cc}
                                        (11 +  \sqrt{-7})/8  & (5 -  \sqrt{-7})/8  \\
                                        (3 + \sqrt{-7})/2 & 1 \\
                                      \end{array}
                                    \right)
$$
$$
=\frac{k_1}{\sqrt{3}} \left(
   \begin{array}{c}
     (5-\sqrt{-7})/8 \\
     1 \\
   \end{array}
 \right)\otimes
 \left(
   \begin{array}{c}
     (3+\sqrt{-7})/2 \\
     1 \\
   \end{array}
 \right).
$$

Step 3:
Assume that $\mathcal{H}$ has a positive Hermitian decomposition
	\begin{equation}\label{EQ:ex6PHD}
	\mathcal{H} =\sum_{k=1}^r p_k \mathcal{A}_k \otimes\mathcal{A}_k^*.
	\end{equation}
Then $r=2$, and one may take
$$
\mathcal{A}_1=\left(
                                      \begin{array}{cc}
                                        (11 -  \sqrt{-7})/8  & (5 +  \sqrt{-7})/8  \\
                                        (3 - \sqrt{-7})/2 & 1 \\
                                      \end{array}
                                    \right),
\mathcal{A}_2=\left(
                                      \begin{array}{cc}
                                        (11 +  \sqrt{-7})/8  & (5 -  \sqrt{-7})/8  \\
                                        (3 + \sqrt{-7})/2 & 1 \\
                                      \end{array}
                                    \right).
$$

Step 4:
Compute partial entries of $\mathcal{H}$ in (\ref{ex:mixedPHD}) and (\ref{EQ:ex6PHD}) as the following table, respectively.

\begin{center}\begin{tabular}{|c|c|c|c|c|}
  \hline
  $\mathcal{H}$                      & $\mathcal{H}_{1111}$ & $\mathcal{H}_{1212}$ & $\mathcal{H}_{2121}$ & $\mathcal{H}_{2222}$ \\ \hline 
  $\mathcal{H}$ of (\ref{ex:mixedPHD}) & $(6 \rho_1 + \rho_2)/18$ & $(6 \rho_1 + \rho_2)/18$ & $(8 \rho_2)/9$ & $\rho_1/3$ \\
  $\mathcal{H}$ of (\ref{EQ:ex6PHD}) & $2 (p_1 + p_2)$ & $(p_1 + p_2)/2$ & $4 (p_1 + p_2)$ & $p_1 + p_2$ \\
  \hline
\end{tabular}\end{center}
Comparing entries of two $\mathcal{H}$, we find that there are no $p_1$ and $p_2$ such that two $\mathcal{H}$ are the same tensor. Hence, $\mathcal{H}$ is not positive Hermitian decomposable. It follows that $\rho$ is entangled. \ $\Box$

\end{example}

\section{Conclusion}
Hermitian tensor can be seen as an extension of Hermitian matrix to higher order. This paper introduces the concepts of Hermitian tensors, partial traces, rank-one Hermitian decomposition, Hermitian tensor eigenvalues and positive Hermitian tensors, etc, and gives their basic properties. All these concepts are useful in quantum physics. A fundamental problem in quantum physics and also an important problem in quantum information science is to detect whether a given state is separable or entangled, and if so, how entangled it is. Hence, based on the consideration of studying on entanglement of quantum mixed states, there are many aspects that need to be studied in the future, including: (1) discrimination of positive Hermitian tensors and decomposition algorithms; (2) numerical methods of calculating quantum entanglement value; (3) properties of partial traces of Hermitian tensors.

 \end{document}